\documentclass[english]{revtex4}
\usepackage[T1]{fontenc}
\usepackage[latin9]{inputenc}
\usepackage{verbatim}
\usepackage{amsthm}
\usepackage{amsmath}
\usepackage{amssymb}
\usepackage{graphicx}

\makeatletter
\@ifundefined{textcolor}{}
{%
 \definecolor{BLACK}{gray}{0}
 \definecolor{WHITE}{gray}{1}
 \definecolor{RED}{rgb}{1,0,0}
 \definecolor{GREEN}{rgb}{0,1,0}
 \definecolor{BLUE}{rgb}{0,0,1}
 \definecolor{CYAN}{cmyk}{1,0,0,0}
 \definecolor{MAGENTA}{cmyk}{0,1,0,0}
 \definecolor{YELLOW}{cmyk}{0,0,1,0}
}
  \theoremstyle{plain}
  \newtheorem{lem}{\protect\lemmaname}
  \newtheorem{theorem}{Theorem}

\usepackage{babel}

\usepackage{babel}
\providecommand{\lemmaname}{Lemma}

\makeatother

\usepackage{babel}
  \providecommand{\lemmaname}{Lemma}

\begin{document}

\title{Unfrustration Condition and Degeneracy of Qudits on Trees}

\author{Matthew Coudron}

\email{mcoudron@mit.edu}

\selectlanguage{english}%

\affiliation{Department of Electrical Engineering and Computer Science, Massachusetts
Institute of Technology, Cambridge MA 02139}

\author{Ramis Movassagh}

\email{ramis.mov@gmail.com}

\selectlanguage{english}%

\affiliation{Department of Mathematics, Northeastern University,
Boston MA  02115}

\date{\today}
\begin{abstract}
We generalize the previous results of \cite{RS} by proving unfrustration condition and degeneracy of the ground states of qudits ($d-$dimensional spins) on a $k-$child tree with generic local interactions. We find that the dimension of the ground space grows doubly exponentially in the region where $rk\leq\frac{d^2}{4}$ for $k>1$.  Further, we extend the results in \cite{RS} by proving that there are no zero energy ground states when $r>\frac{d^2}{4}$ for $k=1$ implying that the effective Hamiltonian is invertible.  
\end{abstract}
\maketitle
\section{Frustration Free Spin Systems }
The interactions in quantum many-body systems are usually well approximated to be local. We say the ground state of the Hamiltonian is {\em unfrustrated} or Frustration Free (FF) when it is also a common ground state of all of the local terms.

 There are many models such as the Heisenberg ferromagnetic chain, AKLT, parent Hamiltonians of MPS that are FF \cite{Koma95, AKLT87, fannes, MPS}. Besides such models and the mathematical convenience of working with FF systems, what is the significance of FF systems? In particular, do FF systems describe systems that can be realized in nature? Some answers can be given.
 
 It has been proved by Hastings \cite{Hastings06} that gapped Hamiltonians can be approximated by FF Hamiltonians if one allows for the range of interaction to be $\mathcal{O}\mathcal{\left(\log N\right)}$.  Further, it is believed that any type of gapped ground state is adequately described by a FF model \cite{fannes}.  A nice feature of FF systems is that the ground state is stable against variation of the Hamiltonian against perturbations  $H\left(g\right)=\sum_{k}g_{k}H_{k,k+1},  g_{k}>0$ as the kernels of the local terms remain invariant \cite{krausMarkov}. In complexity theory, the classical SAT problem was generalized by Bravyi \cite{bravyi} to the so called quantum SAT or qSAT. The statement of the qSAT problem is: Given a collection of $m-$local projectors on $n$  qubits, is there a state $\psi\rangle$ that is annihilated by all the projectors? Namely, is the system FF? Lastly, an important physical motivation was given by Verstraete et al \cite{VerstraeteNature} where they showed that  ground states of FF Hamiltonians can be prepared by dissipation.
 
 Previously, Ref. \cite{RS} focused on a chain of $d-$dimensional spins with  `generic'  local Hamiltonian $H=\sum_j H_{j,j+1}$. The local terms  $H_{j,j+1}$ were chosen randomly with a fixed rank $r$. Three  regimes were identified: (i) frustrated chains, $r >d^2/4$, (ii) FF chains, $d\le r\le d^2/4$, and (iii) FF chains with product ground states, $r<d$.
It was conjectured that the ground states of generic FF chains in the regime $d\le r\le d^2/4$, with probability one, are {\it all} highly entangled in a Schmidt rank sense. This regime however requires local dimension $d\ge 4$. 

In this paper we extend the previous work to the case where the spins are on a tree. Moreover, we improve on the previous results \cite{RS}  on a line by proving that  there are no zero energy ground states when $r>d^2/4$ for qudit chains. We leave the problem of entanglement of the ground states open, though we believe the ideas presented herein (e.g., Lemma \ref{columnranklemma-1}) may ultimately, combined with techniques in \cite{RamisThesis}, become helpful in proving lower bounds on the Schmidt rank.

\section{Generic Interaction }

Consider $k$-child trees of $d$-dimensional quantum particles (qudits)
with nearest neighbor interactions- at every vertex, $k$ edges fan
out to connect to $k$ qudits as shown in Fig. \ref{fig:The-tree-structure}. The Hamiltonian of the system, 
\begin{equation}
H=\sum_{\langle m,n\rangle}^{N-1}H_{m,n}\label{eq:hamiltonian}
\end{equation}
 is $2-$local; each $H_{m,n}$, shown as edges in Fig. \ref{fig:The-tree-structure}, acts non-trivially
only on two neighboring qudits . Our goal is to find the necessary and sufficient conditions for the quantum system, with generic local interactions, to be unfrustrated. Namely, the conditions under which ground states of the Hamiltonian are also common ground states of all $H_{m,n}$.

By {\it generic} we mean randomly sampled from any measure that is absolutely continuous with respect to the Haar measure. In this context, our notion of generic means that no particular local projector has a positive probability of being sampled. 

As discussed previously \cite{RS, RamisThesis}, the
question of existence of a common ground state of all the local terms
is equivalent to asking the same question for an {\it effective} Hamiltonian whose
interaction terms are,
\begin{eqnarray}
H'_{m,n} & = & \mathbb{I}\otimes \Pi_{m,n}\otimes\mathbb{I},\label{eq:projham}
\end{eqnarray}
 with $\Pi_{m,n}$ projecting onto the excited
states of each original interaction term $H_{m,n}$. When this modified
system is unfrustrated, its ground state energy is zero (all the terms
are positive semi-definite). The unfrustrated ground state belongs
to the intersection of the ground state subspaces of each $H_{m,n}$
and is annihilated by all the projector terms.

We choose to focus on a class of Hamiltonians for which each 
\begin{eqnarray}
\Pi_{m,n}=\sum_{p=1}^{r}|v_{m,n}^{p}\rangle\langle v_{m,n}^{p}|\label{eq:definev}
\end{eqnarray}
 is a rank-$r$ projector acting on a $d^{2}$-dimensional Hilbert
space of two qudits, chosen by picking an orthonormal set of $r$
random vectors without translational invariance. %

The set of $r$ constraints of each local term can be seen as a $d^2 \times r$ matrix whose columns are the orthonormal vectors $|v_{m,n}^{p}\rangle$. This matrix is represented by a point on the Steifel manifold \cite{alan}.

\section{Recursive Investigation of Unfrustrated Ground States}

We now find conditions governing the existence of zero energy ground
states (from now on, called solutions in short). We do so by counting
the number of solutions possible for a subset of the tree, and then
adding another site and imposing the constraints given by the Hamiltonian.

\begin{figure}
\begin{centering}
\includegraphics[scale=0.5]{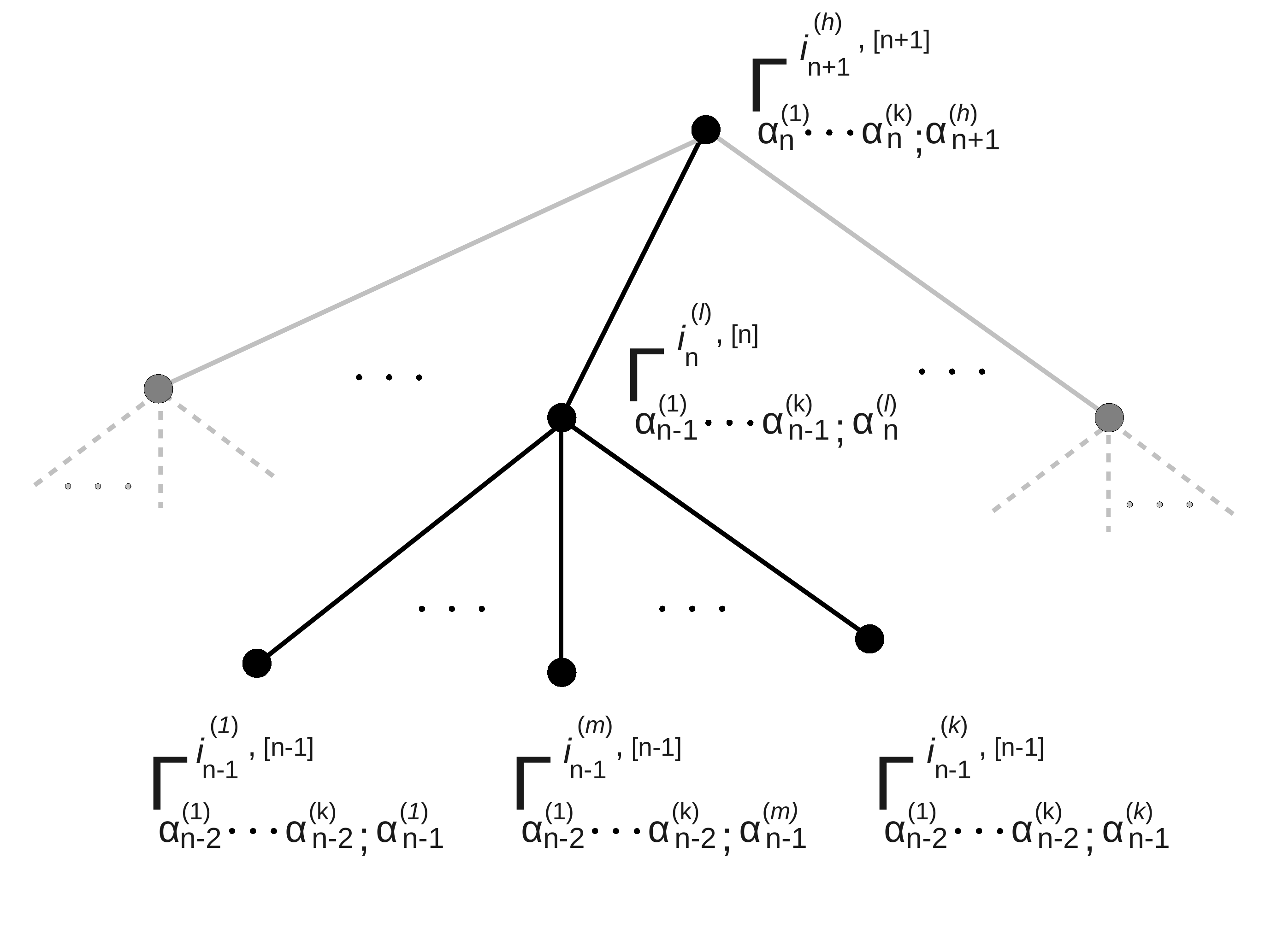}
\par\end{centering}

\caption{\label{fig:The-tree-structure}The tree structure with relevant indexing}
\end{figure}

Below we use the extension of matrix product states (MPS) representation \cite{MPS,vidal1,vidal2,Daniel} to describe the state of the qudits on the tree (also known as tensor product states) \cite{Daniel, vidalTree}. The structure of every tensor (as in MPS) at a given site is $\Gamma_{\alpha_{n-1}^{(1)},\alpha_{n-1}^{(2)},\cdots,\alpha_{n-1}^{(k)};\alpha_{n}^{(l)}}^{i_{n}^{(l)},[n]}$,
where the subscript $\alpha_{n}^{(l)}$ indicates the connection with
the parent and  the subscripts $\alpha_{n-1}^{(1)},\alpha_{n-1}^{(2)},\cdots,\alpha_{n-1}^{(k)}$
indicate connections with the $k$ children of that parent. We denote the membership among the $k$-edges by putting the corresponding label (e.g., $1\leq l\leq k$)  in parenthesis as can be seen in Fig. \ref{fig:The-tree-structure}. The value
of $n$ in what follows {\it increases} as we work our way {\it up} from the leaves to the
root.

 For a given $n$ we will focus on a subtree rooted at some node
$h$ on the $(n+1)^{st}$ level of the tree; a distance $(n+1)$ from
the leaves. We will solve for $|\psi_{\alpha_{n+1}}^{(h)}\rangle$
which represent all of the linearly independent unfrustrated solutions
on the entire subtree rooted at $h$. We will assume inductively that
we have enumerated all of the linearly independent solutions $|\psi_{\alpha_{n}}^{(j)}\rangle$
on subtrees at the $n^{th}$ level for $1\le j\le k$ (see Fig.
 \ref{fig:The-tree-structure}).

By definition the values $\Gamma_{\alpha_{n}^{(1)},\alpha_{n}^{(2)},\cdots,\alpha_{n}^{(k)};\alpha_{n+1}^{(h)}}^{i_{n+1}^{(h)},[n]}$
give the coefficients of the expansion of $|\psi_{\alpha_{n+1}}^{(h)}\rangle$

\begin{eqnarray}
|\psi_{\alpha_{n+1}}^{(h)}\rangle & = & \sum_{i_{n+1}^{(h)}}\sum_{\alpha_{n}^{(1)},\alpha_{n}^{(2)},\cdots,\alpha_{n}^{(k)}}\Gamma_{\alpha_{n}^{(1)},\alpha_{n}^{(2)},\cdots,\alpha_{n}^{(k)};\alpha_{n+1}^{(h)}}^{i_{n+1}^{(h)},[n+1]}|i_{n+1}^{(h)}\rangle|\psi_{\alpha_{n}^{(1)}}^{(1)}\rangle\cdots|\psi_{\alpha_{n}^{(k)}}^{(k)}\rangle,\label{gammaexp1}
\end{eqnarray}
 in terms of the physical index $i_{n+1}^{\left(h\right)}$ and independent
bases for each subtree.

In order to solve for the unfrustration condition and the degeneracy
of the ground states on the subtree rooted at $h$ we must apply the
constraints associated to the projectors between $h$ and each of
its children. For a given child $l$ of $h$ it follows from the unfrustration
condition that $|\psi_{\alpha_{n+1}}^{(h)}\rangle$ must be annihilated
by (see the corresponding edge in Figure \eqref{fig:The-tree-structure})

\[
H'_{n^{(l)},(n+1)^{(h)}}=\mathbb{I}\otimes \Pi_{n^{(l)},(n+1)^{(h)}}\otimes\mathbb{I},
\]
 where

\[
\Pi_{n^{(l)},(n+1)^{(h)}}=\sum_{p=1}^{r}|v_{n^{(l)},(n+1)^{(h)}}^{p}\rangle\langle v_{n^{(l)},(n+1)^{(h)}}^{p}|.
\]
 Here $|v_{n^{(l)},(n+1)^{(h)}}^{p}\rangle$ is a set of random orthonormal
vectors drawn from the $d^{2}$ dimensional space spanned by $|i_{n}^{(l)},i_{n+1}^{(h)}\rangle$.
Clearly, the unfrustration condition implies that  $|\psi_{\alpha_{n+1}}^{(h)}\rangle$
is annihilated by each one of the rank-$1$ projectors, 

\[
H''^{\;p}_{n^{(l)},(n+1)^{(h)}}\equiv\mathbb{I}\otimes|v_{n^{(l)},(n+1)^{(h)}}^{p}\rangle\langle v_{n^{(l)},(n+1)^{(h)}}^{p}|\otimes\mathbb{I}
\]
 for every $p$. Using

\begin{eqnarray}
|\psi_{\alpha_{n}}^{(l)}\rangle & = & \sum_{i_{n}^{(l)},\alpha_{n-1}^{(1)},\alpha_{n-1}^{(2)},\cdots,\alpha_{n-1}^{(k)}}\Gamma_{\alpha_{n-1}^{(1)},\alpha_{n-1}^{(2)},\cdots,\alpha_{n-1}^{(k)};\alpha_{n}^{(l)}}^{i_{n}^{(l)},[n]}|i_{n}^{(l)}\rangle|\psi_{\alpha_{n-1}^{(1)}}^{(1)}\rangle\cdots|\psi_{\alpha_{n-1}^{(k)}}^{(k)}\rangle\label{gammaexp2}
\end{eqnarray}
 and combining it with Eq. \eqref{gammaexp1} we get the expression

\begin{eqnarray}
|\psi_{\alpha_{n+1}}^{(h)}\rangle & = & \sum_{i_{n}^{(l)},i_{n+1}^{(h)}}\sum_{\alpha_{n}^{(1\cdots k)},\alpha_{n-1}^{(1\cdots k)}}\Gamma_{\alpha_{n-1}^{(1\cdots k)};\alpha_{n}^{(l)}}^{i_{n}^{(l)},[n]}\Gamma_{\alpha_{n}^{(1\cdots k)};\alpha_{n+1}^{(h)}}^{i_{n+1}^{(h)},[n+1]}|i_{n}^{(l)},i_{n+1}^{(h)}\rangle|\psi_{\alpha_{n}^{(1)}}^{(1)}\rangle\cdots\widehat{|\psi_{\alpha_{n}^{(l)}}^{(l)}\rangle}\cdots|\psi_{\alpha_{n}^{(k)}}^{(k)}\rangle|\psi_{\alpha_{n-1}^{(1)}}^{(1)}\rangle\cdots|\psi_{\alpha_{n-1}^{(k)}}^{(k)}\rangle,\label{gammaexp3}
\end{eqnarray}
 where we denote $\alpha_{n-1}^{(1)},\alpha_{n-1}^{(2)},\cdots,\alpha_{n-1}^{(k)}$ by $\alpha_{n-1}^{(1\cdots k)}$ and a missing quantity by  an  $\;\;\widehat{ }\;\;$  over that quantity.

We now apply the projector and consider its kernel

\begin{eqnarray}
H''^{\;p}_{n^{(l)},(n+1)^{(h)}}|\psi_{\alpha_{n+1}}^{(h)}\rangle & = & \left\{ \mathbb{I}\otimes|v_{n^{(l)},(n+1)^{(h)}}^{p}\rangle\langle v_{n^{(l)},(n+1)^{(h)}}^{p}|\otimes\mathbb{I}\right\} |\psi_{\alpha_{n+1}}^{(h)}\rangle\\
 & = & \sum_{i_{n}^{(l)},i_{n+1}^{(h)}}\sum_{\alpha_{n}^{(1\cdots k)},\alpha_{n-1}^{(1\cdots k)}}\Gamma_{\alpha_{n-1}^{(1\cdots k)};\alpha_{n}^{(l)}}^{i_{n}^{(l)},[n]}\Gamma_{\alpha_{n}^{(1\cdots k)};\alpha_{n+1}^{(h)}}^{i_{n+1}^{(h)},[n+1]}|v_{n^{(l)},(n+1)^{(h)}}^{p}\rangle\langle v_{n^{(l)},(n+1)^{(h)}}^{p}|i_{n}^{(l)},i_{n+1}^{(h)}\rangle\\
 & \otimes & |\psi_{\alpha_{n}^{(1)}}^{(1)}\rangle\cdots\widehat{|\psi_{\alpha_{n}^{(l)}}^{(l)}\rangle}\cdots|\psi_{\alpha_{n}^{(k)}}^{(k)}\rangle|\psi_{\alpha_{n-1}^{(1)}}^{(1)}\rangle\cdots|\psi_{\alpha_{n-1}^{(k)}}^{(k)}\rangle\\
 & = & 0 .
\end{eqnarray}
The set of vectors $|v_{n^{(l)},(n+1)^{(h)}}^{p}\rangle|\psi_{\alpha_{n}^{(1)}}^{(1)}\rangle\cdots\widehat{|\psi_{\alpha_{n}^{(l)}}^{(l)}\rangle} \cdots |\psi_{\alpha_{n}^{(k)}}^{(k)}\rangle|\psi_{\alpha_{n-1}^{(1)}}^{(1)}\rangle\cdots|\psi_{\alpha_{n-1}^{(k)}}^{(k)}\rangle$
with respect to variables $\alpha_{n-1}^{(1\cdots k)}$ and $\alpha_{n}^{(1\cdots k)}$
are linearly independent. This follows from the fact that $\forall j$,
$|\psi_{\alpha_{n}^{(j)}}^{(j)}\rangle$ are linearly independent
for different values of $\alpha_{n}^{(j)}$, and that $|\psi_{\alpha_{n}^{(j)}}^{(j)}\rangle$
describe states on completely disjoint subtrees for different values
of $j$. Using this linear independence we see that the equation above
is true if and only if (iff)

\begin{align}
& \langle v_{n^{(l)},(n+1)^{(h)}}^{p}|i_{n}^{(l)},i_{n+1}^{(h)}\rangle\Gamma_{\alpha_{n-1}^{(1\cdots k)};\alpha_{n}^{(l)}}^{i_{n}^{(l)},[n]}\Gamma_{\alpha_{n}^{(1\cdots k)};\alpha_{n+1}^{(h)}}^{i_{n+1}^{(h)},[n+1]}=0\nonumber \label{constraint1}\\
\forall & \alpha_{n-1}^{(1)},\cdots,\alpha_{n-1}^{(k)},\quad\text{ and }\nonumber \\
\forall & \alpha_{n}^{(1)},\cdots,\widehat{\alpha_{n}^{(l)}},\cdots,\alpha_{n}^{(k)}.
\end{align}

 Note that in Eq. \eqref{constraint1} the repeated indices $i_{n}^{(l)}$,
$i_{n+1}^{(h)}$ and $\alpha_{n}^{(l)}$ are summed over. This shows
that the unfrustration condition holds iff Eq. \eqref{constraint1}
holds for all $1\leq p\leq r$, and $1\leq l\leq k$.

These constraints may be rewritten as

\begin{eqnarray}
\sum_{i_{n+1}^{(h)}}\sum_{\alpha_{n}^{(1\cdots k)}}C_{p,l,\alpha_{n-1}^{(1\cdots k)},\beta_{n}^{(1)},\cdots,\widehat{\beta_{n}^{(l)}},\cdots,\beta_{n}^{(k)};i_{n+1}^{(h)},\alpha_{n}^{(1\cdots k)}}\Gamma_{\alpha_{n}^{(1\cdots k)};\alpha_{n+1}^{(h)}}^{i_{n+1}^{(h)},[n+1]} & = & 0,\;\;\;  \forall p,l \label{eq:LinearSys}
\end{eqnarray}
where,

\begin{equation}
C_{p,l,\alpha_{n-1}^{(1\cdots k)},\beta_{n}^{(1)},\cdots,\widehat{\beta_{n}^{(l)}},\cdots,\beta_{n}^{(k)};i_{n+1}^{(h)},\alpha_{n}^{(1\cdots k)}}\equiv\prod_{j=1,\cdots,\hat{l,}\cdots,k}\delta\left(\beta_{n}^{j}=\alpha_{n}^{j}\right)\left(\sum_{i_{n}^{(l)}}\langle v_{n^{(l)},(n+1)^{(h)}}^{p}|i_{n}^{(l)},i_{n+1}^{(h)}\rangle\Gamma_{\alpha_{n-1}^{(1\cdots k)};\alpha_{n}^{(l)}}^{i_{n}^{(l)},[n]}\right)\label{eq:Cmatrix}
\end{equation}
In Eq. \eqref{eq:Cmatrix} the dummy variables $\beta_{n}^{(1)},\cdots,\widehat{\beta_{n}^{(l)}},\cdots,\beta_{n}^{(k)}$
have exactly the same ranges of value as $\alpha_{n}^{(1)},\cdots,\widehat{\alpha_{n}^{(l)}},\cdots,\alpha_{n}^{(k)}$
but take values independently. 

Comment: We reserve the notation $\alpha$'s for when the constraints on all the other edges have been satisfied too. 
The delta notation is to emphasize that the constraints must hold
for any choice of $\alpha_{n}^{(1)},\cdots,\widehat{\alpha_{n}^{(l)}},\cdots,\alpha_{n}^{(k)}$,
i.e., for all other subtrees other than $l$. 

The constraint matrix $C_{p,l,\alpha_{n-1}^{(1\cdots k)},\beta_{n}^{(1)},\cdots,\widehat{\beta_{n}^{(l)}},\cdots,\beta_{n}^{(k)};i_{n+1}^{(h)},\alpha_{n}^{(1\cdots k)}}$, also denoted simply by $C$, 
has $dD_{n}^{k}$ columns since $1\le i_{n+1}^{\left(h\right)}\le d$
and $1\le\alpha_{n}^{(1\cdots k)}\le D_{n}^{k}$ . Further it has
$rkD_{n-1}^{k}D_{n}^{k-1}$ rows since $1\le p\le r$, $1\le l\le k$,
$1\le\alpha_{n-1}^{\left(1\cdots k\right)}\le D_{n-1}^{k}$ and $1\le\beta_{n}^{(1)},\cdots,\widehat{\beta_{n}^{(l)}},\cdots,\beta_{n}^{(k)}\le D_{n}^{k-1}$.

Now, if the matrix $C$ has full rank with probability
one and $dD_{n}^{k}>rkD_{n-1}^{k}D_{n}^{k-1}$
for suitable values of $r,k,d$, then the kernel of $C$ has dimension

\begin{align}
D_{n+1} & \equiv dD_{n}^{k}-rkD_{n-1}^{k}D_{n}^{k-1}\quad\mbox{with probability one.}\label{recursion}
\end{align}
 It follows that there are $D_{n+1}$ linearly independent solutions
$\Gamma_{\alpha_{n}^{(1\cdots k)};\alpha_{n+1}^{(h)}}^{i_{n+1}^{(h)},[n+1]}$,
and thus $D_{n+1}$ linearly independent solutions $|\psi_{\alpha_{n+1}}^{(h)}\rangle$
on the subtree rooted at $h$. In the appendix we prove that $C$ is indeed full rank. 

Furthermore, by the same token, if we have that $dD_{n}^{k}\leq rkD_{n-1}^{k}D_{n}^{k-1}$
then there are no solutions $|\psi_{\alpha_{n+1}}^{(h)}\rangle$ on
the subtree rooted at $h$ implying that the Hamiltonian is \textit{frustrated}.
We proceed to analyze the recursion Eq. \eqref{recursion}, determine the
criteria for $r,k$ and $d$ that assure the existence of unfrustrated ground
states, and investigate the asymptotic growth of the number of solutions.

\section{\label{sec:Recursion-Analysis}Recursion Analysis}

Consider the recursion in Eq. (\ref{recursion})
\begin{align}
 & D_{n+1}\equiv dD_{n}^{k}-rkD_{n-1}^{k}D_{n}^{k-1}\label{recursion2}
\end{align}
 with the initial conditions $D_{0}=1$, $D_{1}=d$. Recall that we
start at the leaves of the tree, where each unrestricted qudit on
a leaf lives in a $d$-dimensional Hilbert space.
The value $D_{0}=1$ can be viewed as a place-holder in the recursion
and represents a formal $1-$dimensional space preceding the leaves.

Now suppose the solutions have the form,

\begin{eqnarray*}
D_{n} & = & \gamma_{n}D_{n-1}^{k}
\end{eqnarray*}
 for some $\gamma_{n}\in\mathbb{R}$. It follows from the recursion
\eqref{recursion2} that

\begin{eqnarray}
D_{n+1}=dD_{n}^{k}-rkD_{n-1}^{k}D_{n}^{k-1} & = & \left(d-\frac{rk}{\gamma_{n}}\right)D_{n}^{k}
\end{eqnarray}

Thus, if we define $\gamma_{n}$ by the recursion

\begin{eqnarray}
\gamma_{n+1} & \equiv & \left(d-\frac{rk}{\gamma_{n}}\right)\label{gammarecursion}\\
\gamma_{1} & \equiv & \frac{D_{1}}{D_{0}^{k}}=d;\nonumber 
\end{eqnarray}
 it follows that $D_{n}=\gamma_{n}D_{n-1}^{k}$ $\forall n$. Provided that we have founds positive solutions up to the $n^{th}$ step, i.e., non-negative $D_0,D_1,\cdots,D_{n-1}$, 
the value of $D_{n}$ becomes non-positive iff the value
of $\gamma_{n}$ becomes non-positive.

The following expressions are equivalent:
\begin{eqnarray}
\gamma_{n+1} =  \left(d-\frac{rk}{\gamma_{n}}\right)\geq\gamma_{n}\qquad \Longleftrightarrow \qquad\gamma_{n}^{2}-d\gamma_{n}+rk & \leq & 0. \label{eq:decreasecondition}
\end{eqnarray}
whose roots, taking the equality, are denoted by
\begin{eqnarray}
x_{-} \equiv & \frac{d-\sqrt{d^{2}-4rk}}{2} ,\qquad x_{+}\equiv\frac{d+\sqrt{d^{2}-4rk}}{2}\;.\nonumber 
\end{eqnarray}

Note that the inequality (\ref{eq:decreasecondition}) is satisfied exactly when $\gamma_{n}\in\left[x_{-},x_{+}\right]$. 
In the above computation we are assuming that $\gamma_{n}$
is positive since a non-positive value of $\gamma_{n}$ indicates
that there are no unfrustrated solutions on the chain with $n$ (or
more) sites. When $rk>\frac{d^{2}}{4}$ these roots are not real and it follows
that $\gamma_{n}$ is a strictly decreasing sequence. We thus know
that $\gamma_{n}$ must eventually become non-positive, or it must
converge to a positive number. However, it is easy to see that if
$\gamma_{n}$ converges to some positive number $\gamma^*$ then $\gamma^*$
must be a fixed point of (\ref{eq:decreasecondition}) satisfying $\gamma^{*2}-d\gamma^*+rk=0$, but this is impossible since
the roots are not real. It follows that, in the case $rk>\frac{d^{2}}{4}$
there exists an $N$ such that, for all $n\geq N$ there are no unfrustrated
solutions on the $n$ site chain (with probability $1$).

On the other hand, if $rk\leq\frac{d^{2}}{4}$ we note that if $\gamma_{n}\geq x_{+}$
we have that

\[
\gamma_{n+1}=\left(d-\frac{rk}{\gamma_{n}}\right)\geq\left(d-\frac{rk}{x_{+}}\right)=x_{+}.
\]

Since $\gamma_{0}=d\geq x_{+}$, using Eq. \eqref{eq:decreasecondition},
 $\gamma_{n}$ is a decreasing sequence which is bounded below
by $x_{+}$. Therefore,  $\gamma_{n}$ must converge to some
$\gamma\geq x_{+}$, which implies that its limit $\gamma$
must be a fixed point of $\gamma^{*2}-d\gamma^*+rk=0$, hence
\begin{align}
\lim_{n\to\infty}\gamma_{n}=x_{+}\label{gammalimit}\\
\nonumber 
\end{align}

It follows that, for $rk\leq\frac{d^{2}}{4}$, $D_{n}\geq G_{n}$  where $G_{n}$ is the solution to the recursion $G_{n}=x_{+}G_{n-1}^{k}$
with $G_{1}=d$ and
\begin{align}
D_{n}\geq G_{n} & \equiv x_{+}^{s_{n-1}}d^{k^{n-1}},\label{lowerbound}
\end{align}
where $s_{n-1}\equiv\sum{_{l=0}^{n-2}k^l}=\frac{k^{n-1}-1}{k-1}$. Eq. (\ref{lowerbound}) for all $d\ge2\Rightarrow x_+\ge 1$  implies a growing number of solutions.

Furthermore, the recursion for $G_{n}$ and that
$D_{n}$ converge in the sense that the respective recursion
constants $\gamma_{n}$ and $x_{+}$ converge. In particular, Eq.
\eqref{lowerbound} shows that, in the regime $rk\leq\frac{d^{2}}{4}$
the dimension of the unfrustrated ground space grows doubly exponentially as long as $k>1$.

\section{Proof of Frustration for $k=1$}
We now prove the non-existence of unfrustrated ground states for the $n$-qudit Hamiltonian with generic local interactions on the line when $r>\frac{d^2}{4}$.  In \cite{RS} the unfrustration condition was proved; however, it was only conjectured that the kernel would be empty with probability one when $r>\frac{d^2}{4}$. Naturally, the result below holds for sufficiently large $n$ since when $n$ is small the Hamiltonian may have zero eigenvalues.  

The intuition for the (non-)existence of the zero energy ground states follows from the solution of the recursion relation in Eq. (\ref{recursion}). It follows from sections II and III that the dimension of unfrustrated ground states is given by the solution of the recursion relation Eq. (\ref{recursion}) as long as $D_n$ is non-negative.  We also know from section III that $r \leq \frac{d^2}{4}$ implies that $D_n \geq 0$ $\forall n$, and that $r > \frac{d^2}{4}$ implies $D_n \leq 0$ for some $n$.  It is natural to conjecture that the Hamiltonian is frustrated in the regime $r > \frac{d^2}{4}$ for sufficiently large $n$.  

 We define $E_{n}$ to be the dimension of the kernel of the Hamiltonian on the first $n$ qudits, which we distinguish from $D_n$. The latter being the solution to the recursion Eq. (\ref{recursion}).  Of course we still have $E_n = D_n$ for sufficiently small $n$.  In this section we prove that, when $\frac{d^2}{2} \geq r > \frac{d^2}{4}$, $D_{n_0+1} \leq 0$ implies $E_{n_0+2} = 0$; i.e., the chain becomes frustrated.  Note that the restriction $\frac{d^2}{2} \geq r$ may be used without loss of generality (WLOG) since non-existence of unfrustrated states when $\frac{d^2}{2} \geq r > \frac{d^2}{4}$ automatically implies non-existence of unfrustrated states when $r > \frac{d^2}{2}$.  

We recall that $D_{n_0+1} = d D_{n_0} - r D_{n_0-1}$, so that $D_{n_0+1} \leq 0$ iff $\frac{D_{n_0}}{D_{n_0-1}} \leq \frac{r}{d}$.  Thus, we would like to start with this second condition and prove the desired result.  We begin with a lemma which gives us the desired result, but uses a slightly stronger condition.  

\begin{lem}
\label{k1} Assume that $n_0 \in \mathbb{Z}^+$ is such that $E_n = D_n > 0$ for $n \leq n_0-1$, $E_{n_0} >0$, and that $\lceil \frac{E_{n_0}}{E_{n_0-1}} \rceil \leq \frac{r}{d}$.  Then $E_{n_0+1} = 0$  with probability one.
\end{lem}

\begin{proof}
For $k=1$ the constraint matrix is $C_{p,\alpha_{n_0-1};i_{n_0+1},\alpha_{n_0}}$, which has $r E_{n_0-1}$ rows, and $d E_{n_0}$ columns by definition.  It follows from the assumption that $d E_{n_0} \leq r E_{n_0-1}$, so $C$ has more rows than columns and one needs to prove linear independence of the columns in order to prove that the kernel is empty.  Thus, we must prove that the statement 

\begin{equation}
\sum_{i_{n_0+1}, \alpha_{n_0}} y_{i_{n_0+1}, \alpha_{n_0}} C_{p,\alpha_{n_0-1}; i_{n_0+1}, \alpha_{n_0} } =  \sum_{i_{n_0+1}, \alpha_{n_0}} y_{i_{n_0+1}, \alpha_{n_0}} \left(\sum_{i_{n_0}}\langle v_{n_0,(n_0+1)}^{p}|i_{n_0},i_{n_0+1}\rangle\Gamma_{\alpha_{n_0-1};\alpha_{n_0}}^{i_{n_0},[n_0]}\right) = 0 \quad
\forall p,\alpha_{n_0-1} \label{columnlinindep}
\end{equation}
implies

\begin{equation}
y_{i_{n_0+1}, \alpha_{n_0}} = 0 \text{   } \forall  \text{ } i_{n_0+1}, \alpha_{n_0}.
\end{equation}

Following the reasoning in the Appendix, we know that (WLOG, and with probability 1) we may apply Lemma \ref{columnranklemma-1} to row reduce the matrix $\Gamma_{\alpha_{n_0-1};\alpha_{n_0}}^{i_{n_0},[n_0]}$ on the set of rows

\begin{align}
& s \equiv \left \{ (i_{n_0}, \alpha_{n_0-1}) :  i_{n_0} \in \left [1, ..., \left \lfloor \frac{E_{n_0}}{E_{n_0-1}} \right \rfloor   \right ] \text{, } \alpha_{n_0-1} \in  \left [1, ..., E_{n_0-1} \right ]  \right \} \label{GammaRed}\\ 
& \cup  \left \{ (i_{n_0}, \alpha_{n_0-1}) :  i_{n_0} =  \left \lceil \frac{E_{n_0}}{E_{n_0-1}} \right \rceil  \text{, }  \alpha_{n_0-1} \in  \left [1, ..., E_{n_0}  - E_{n_0-1} \cdot \left \lfloor \frac{E_{n_0}}{E_{n_0-1}} \right \rfloor \right ] \right \} \nonumber
\end{align}
Note, in particular that $|s| = E_{n_0}$, and for $(i_{n_0}, \alpha_{n_0-1}) \in s$ we have $i_{n_0} \leq \left \lceil \frac{E_{n_0}}{E_{n_0-1}} \right \rceil \leq \frac{r}{d} \leq \frac{d}{2}$ since $r \leq \frac{d^2}{2}$ by assumption.  Thus we have satisfied the requirements of Lemma \ref{columnranklemma-1} and may assume WLOG that $\Gamma_{\alpha_{n_0-1};\alpha_{n_0}}^{i_{n_0},[n_0]}$ is row reduced on the rows corresponding to $s$.

It follows that, given $\alpha_{n_0}'$, $\exists (i_{n_0}', \alpha_{n_0-1}') \in s$, such that

\begin{eqnarray}
\Gamma_{\alpha_{n_0-1};\alpha_{n_0}'}^{i_{n_0},[n_0]} = \left\{
\begin{array}{c l}      
    1 & \text{ if } (i_{n_0}, \alpha_{n_0-1}) = (i_{n_0}', \alpha_{n_0-1}') \\ 
    0 & \text{ if } (i_{n_0}, \alpha_{n_0-1}) \in s \text{ and } (i_{n_0}, \alpha_{n_0-1}) \neq (i_{n_0}', \alpha_{n_0-1}') .    
\end{array}\right. 
\label{gamma_diag}
\end{eqnarray}

Similarly, we know from the ``geometrization theorem'' of \cite{L} that we only need to prove full rankness of columns of $C_{p,\alpha_{n_0-1};i_{n_0+1},\alpha_{n_0}}$ for a specific choice of projectors.  It will then hold with probability 1 for random projectors.

We will assign projectors as follows:

\begin{eqnarray}
\langle v_{n_0,(n_0+1)}^{p}|i_{n_0},i_{n_0+1}\rangle= \left\{
\begin{array}{c l}      
    1 & \text{ if } i_{n_0} =  \left \lfloor \frac{p}{d} \right \rfloor+1  \text{ and } i_{n_0+1} = p - d \left \lfloor \frac{p}{d} \right \rfloor  \\ 
    0 & \text{ otherwise } .    
\end{array}\right. 
\label{projectorfill}
\end{eqnarray}

Now, given Eq. \eqref{columnlinindep} we will show 

\begin{equation}
y_{i_{n_0+1}, \alpha_{n_0}} = 0 \text{   } \forall  \text{ } i_{n_0+1}, \alpha_{n_0}.
\end{equation}

Given $(i_{n_0+1}', \alpha_{n_0}')$ we choose $(i_{n_0}', \alpha_{n_0-1}')$ corresponding to $\alpha_{n_0}'$ in Eq. \eqref{gamma_diag}.  We then choose $p'$ so that 

\begin{eqnarray}
\langle v_{n_0,(n_0+1)}^{p'}|i_{n_0},i_{n_0+1}\rangle= \left\{
\begin{array}{c l}      
    1 & \text{ if } (i_{n_0},i_{n_0+1}) = (i_{n_0}',i_{n_0+1}')  \\ 
    0 & \text{ otherwise } .    
\end{array}\right. 
\end{eqnarray}

We know that such a $p'$ exists in the range $p \in [1, ..., r]$ because we know $i_{n_0}' \leq \left \lceil \frac{E_{n_0}}{E_{n_0-1}} \right \rceil \leq \frac{r}{d}$ by assumption, so we have $d i_{n_0}' \leq r$.  The existence of such a $p'$ now follows from Eq. \eqref{projectorfill}.  

Eq. \eqref{columnlinindep} now collapses as follows 

\begin{align}
& 0 = \sum_{i_{n_0+1}, \alpha_{n_0}} y_{i_{n_0+1}, \alpha_{n_0}} C_{p',\alpha_{n_0-1}'; i_{n_0+1}, \alpha_{n_0} } =  \sum_{i_{n_0+1}, \alpha_{n_0}} y_{i_{n_0+1}, \alpha_{n_0}} \left(\sum_{i_{n_0}}\langle v_{n_0,(n_0+1)}^{p'}|i_{n_0},i_{n_0+1}\rangle\Gamma_{\alpha_{n_0-1}';\alpha_{n_0}}^{i_{n_0},[n_0]}\right)  \label{collapse} \\  
& =   \sum_{ \alpha_{n_0}} y_{i_{n_0+1}', \alpha_{n_0}} \Gamma_{\alpha_{n_0-1}';\alpha_{n_0}}^{i_{n_0}',[n_0]} = y_{i_{n_0+1}', \alpha_{n_0}'}  \nonumber
\end{align}

Since $(i_{n_0+1}', \alpha_{n_0}')$ was arbitrary we have now proved 

\begin{equation}
y_{i_{n_0+1}, \alpha_{n_0}} = 0 \text{   } \forall  \text{ } i_{n_0+1}, \alpha_{n_0}
\end{equation}

so that the desired result follows.

\end{proof}

Note, in Lemma \ref{k1} that, given $E_{n_0+1} = 0$ it follows easily from the the definition of $E_n$, that $E_n = 0$  $\forall$ $n \geq n_0+1$.

\bigskip

Now, as discussed earlier, we would like to be able to prove that $E_{n_0+1} = 0$ using only the condition $\frac{E_{n_0}}{E_{n_0-1}} \leq \frac{r}{d}$.  However, Lemma \ref{k1} uses the assumption $\lceil \frac{E_{n_0}}{E_{n_0-1}} \rceil \leq \frac{r}{d}$ which is slightly stronger.  We can work around this using the following strategy:  instead of proving $E_{n_0+1} = 0$, we use reasoning similar to that in Lemma \ref{k1} to show that $E_{n_0+1}$ is fairly small.  The bound on $E_{n_0+1}$ will be sufficient to show that $\lceil \frac{E_{n_0+1}}{E_{n_0}} \rceil \leq \frac{r}{d}$  (when $d \geq 8$) and then we can apply Lemma \ref{k1} to show that $E_{n_0+2} = 0$.  This intuition is made precise in the following theorem.

\begin{theorem}
\label{k1p} If $\lceil{\frac{d^2}{2r}}\rceil<\frac{r}{d}$, then the Hamiltonian for qudits on the line with generic local interactions is frustrated for sufficiently large $n$ with probability one. 
\end{theorem}

Comment: $r > \frac{d^2}{4}$ and $d \geq 8$ together imply $\lceil{\frac{d^2}{2r}}\rceil<\frac{r}{d}$, so that this theorem is always valid when $r > \frac{d^2}{4}$, and $d \geq 8$.

\begin{proof}
Assume that $n_0 \in \mathbb{Z}^+$ is such that $E_n = D_n > 0$ for $n \leq n_0$, $D_{n_0} = E_{n_0} >0$, and $D_{n_0+1} \leq 0$, so that $\frac{D_{n_0}}{D_{n_0-1}}  \leq \frac{r}{d}$.  There are now two cases.  If we have that $\lceil \frac{D_{n_0}}{D_{n_0-1}} \rceil  \leq \frac{r}{d}$, then we can apply Lemma \ref{k1} directly to show that $E_{n_0+1} = 0$.  This, in turn implies that $E_n = 0$ for all $n \geq n_0+1$ so that we have $E_{n_0+2} = 0$, and we are done.

In the second case, we have $\lceil \frac{D_{n_0}}{D_{n_0-1}} \rceil  > \frac{r}{d}$.  Since we know $\frac{D_{n_0}}{D_{n_0-1}}  \leq \frac{r}{d}$, we also have $\lfloor \frac{D_{n_0}}{D_{n_0-1}} \rfloor \leq \frac{r}{d}$.  Our goal now is to show that $E_{n_0+1}$ is small by showing that a large subset of the columns of $C_{p,\alpha_{n_0-1}; i_{n_0+1}, \alpha_{n_0} }$ are linearly indepedent.  We will accomplish this by following the general idea behind the proof of Lemma \ref{k1}, except that the role of $\lceil \frac{D_{n_0}}{D_{n_0-1}} \rceil = \lceil \frac{E_{n_0}}{E_{n_0-1}} \rceil$ will be replaced by $\lfloor \frac{E_{n_0}}{E_{n_0-1}} \rfloor$.  As a result we will not be able to prove that $C_{p,\alpha_{n_0-1}; i_{n_0+1}, \alpha_{n_0} }$ has full column rank, but we will select a subset $F$ of the pairs $(i_{n_0+1}, \alpha_{n_0})$, and prove linear indepedence for the corresponding columns of $C_{p,\alpha_{n_0-1}; i_{n_0+1}, \alpha_{n_0} }$ (that is, only for those columns whose labels are contained in $F$).

We first recall the fact that we may, WLOG, use Lemma \ref{columnranklemma-1} to row reduce the matrix $\Gamma_{\alpha_{n_0-1};\alpha_{n_0}}^{i_{n_0},[n_0]}$ on the set $s$ given by Eq. \ref{GammaRed}.

Note, in particular that $|s| = E_{n_0}$.  Further note that, for $(i_{n_0}, \alpha_{n_0-1}) \in s$, we have $i_{n_0} \leq \left \lceil \frac{E_{n_0}}{E_{n_0-1}} \right \rceil \leq \frac{d}{2}$.  The statement $\left \lceil \frac{E_{n_0}}{E_{n_0-1}} \right \rceil \leq \frac{d}{2}$ follows because we assume $r \leq \frac{d^2}{2}$ WLOG (just as in Lemma \ref{k1}), and we have one of two cases.  Either $\frac{E_{n_0}}{E_{n_0-1}}  \leq \frac{r}{d} < \frac{d}{2}$ so that $\left \lceil \frac{E_{n_0}}{E_{n_0-1}} \right \rceil \leq \frac{E_{n_0}}{E_{n_0-1}}  +1 \leq \frac{d}{2}$, or $\frac{E_{n_0}}{E_{n_0-1}}  \leq \frac{r}{d} < \frac{d}{2}$, in which case $\left \lceil \frac{E_{n_0}}{E_{n_0-1}} \right \rceil = \frac{E_{n_0}}{E_{n_0-1}}  \leq \frac{d}{2}$.  Thus, we have satisfied the requirements of Lemma \ref{columnranklemma-1} and may assume WLOG that $\Gamma_{\alpha_{n_0-1};\alpha_{n_0}}^{i_{n_0},[n_0]}$ is row reduced on the rows corresponding to $s$.

It follows that, given $\alpha_{n_0}'$, $\exists (i_{n_0}', \alpha_{n_0-1}') \in s$, such that

\begin{eqnarray}
\Gamma_{\alpha_{n_0-1};\alpha_{n_0}'}^{i_{n_0},[n_0]} = \left\{
\begin{array}{c l}      
    1 & \text{ if } (i_{n_0}, \alpha_{n_0-1}) = (i_{n_0}', \alpha_{n_0-1}') \\ 
    0 & \text{ if } (i_{n_0}, \alpha_{n_0-1}) \in s \text{ and } (i_{n_0}, \alpha_{n_0-1}) \neq (i_{n_0}', \alpha_{n_0-1}') .    
\end{array}\right. 
\label{gamma_diag2}
\end{eqnarray}

Since $|s| = E_{n_0}$ it follows that, given $\alpha_{n_0}'$, the corresponding $(i_{n_0}', \alpha_{n_0-1}')$ is unique.

Similarly, we know from the ``geometrization theorem'' of \cite{L} that we only need to prove full rankness of columns of $C_{p,\alpha_{n_0-1};i_{n_0+1},\alpha_{n_0}}$ for a specific choice of projectors.  It will then hold with probability 1 for random projectors.

We will assign projectors exactly as in Lemma \ref{k1}:

\begin{eqnarray}
\langle v_{n_0,(n_0+1)}^{p}|i_{n_0},i_{n_0+1}\rangle= \left\{
\begin{array}{c l}      
    1 & \text{ if } i_{n_0} =  \left \lfloor \frac{p}{d} \right \rfloor+1  \text{ and } i_{n_0+1} = p - d \left \lfloor \frac{p}{d} \right \rfloor  \\ 
    0 & \text{ otherwise } .    
\end{array}\right. 
\label{projectorfill2}
\end{eqnarray}

We will say that a value $p$, and a tuple $(i_{n_0}, i_{n_0+1})$ are associated if $i_{n_0} =  \left \lfloor \frac{p}{d} \right \rfloor+1  \text{ and } i_{n_0+1} = p - d \left \lfloor \frac{p}{d} \right \rfloor $.  Since there are $r$ projectors we know that $p \in [1, ..., r ]$.  

Recall that $F$ is the set of column labels for the columns of $C_{p,\alpha_{n_0-1}; i_{n_0+1}, \alpha_{n_0} }$ that we wish to prove are linearly independent.  In order to apply the argument of Lemma \ref{k1} we need that, for every $(i_{n_0+1}', \alpha_{n_0}') \in F$, there exists $p' \in [1, ..., r ]$ such that $\langle v_{n_0,(n_0+1)}^{p}|i_{n_0}',i_{n_0+1}'\rangle = 1$  (here $i_{n_0}'$ is the coordinate of the tuple $(i_{n_0}', \alpha_{n_0-1}')$ corresponding to $\alpha_{n_0}'$ via \eqref{gamma_diag2}).  In other words, we need that there exists a value $p' \in [1, ..., r ]$ that is associated with the tuple $(i_{n_0}',i_{n_0+1}')$.

From Eq. \eqref{projectorfill2} we see that, since $d \left \lceil \frac{E_{n_0}}{E_{n_0-1}} \right \rceil \geq  r \geq d \frac{E_{n_0}}{E_{n_0-1}} \geq d \left \lfloor \frac{E_{n_0}}{E_{n_0-1}} \right \rfloor$, the only time the above conditions could fail is when $i_{n_0}'  = \left \lceil \frac{E_{n_0}}{E_{n_0-1}} \right \rceil$.  This follows because, for $i_{n_0}  < \left \lceil \frac{E_{n_0}}{E_{n_0-1}} \right \rceil$, there is always a value of $p \in [1, ..., r ]$ associated to the tuple $(i_{n_0}, i_{n_0+1})$ regardless of the value of $i_{n_0+1}$.  

It follows from Eq. \eqref{gamma_diag2} that there are exactly $E_{n_0} - \left \lfloor \frac{E_{n_0}}{E_{n_0-1}} \right \rfloor E_{n_0-1}$ values of $\alpha_{n_0}'$ such that the corresponding $(i_{n_0}', \alpha_{n_0-1}')$ has $i_{n_0}' = \left \lceil \frac{E_{n_0}}{E_{n_0-1}} \right \rceil$.  As discussed above, only columns with labels $(i_{n_0+1}', \alpha_{n_0}')$ containing such an $\alpha_{n_0}'$ must be excluded from the set $F$ of linearly indepedent columns.  In fact, we need not exclude quite so many.  Since  $r \geq d \frac{E_{n_0}}{E_{n_0-1}}$, it follows from Eq. \eqref{projectorfill2} that, if $i_{n_0} = \left \lceil \frac{E_{n_0}}{E_{n_0-1}} \right \rceil$, and $i_{n_0+1} \leq d \left ( \frac{E_{n_0}}{E_{n_0-1}} - \left \lfloor \frac{E_{n_0}}{E_{n_0-1}} \right \rfloor  \right )$, then there is a still a $p \in [1, ..., r ]$ which is associated with $(i_{n_0}, i_{n_0+1})$.  Thus, we only have to remove a tuple $(i_{n_0+1}', \alpha_{n_0}')$ from $F$ when $\alpha_{n_0}'$ is such that $i_{n_0}' = \left \lceil \frac{E_{n_0}}{E_{n_0-1}} \right \rceil$, and $i_{n_0+1} > d \left ( \frac{E_{n_0}}{E_{n_0-1}} - \left \lfloor \frac{E_{n_0}}{E_{n_0-1}} \right \rfloor  \right )$.  It follows that we only need to remove 

$$\left(d - d \left ( \frac{E_{n_0}}{E_{n_0-1}} - \left \lfloor \frac{E_{n_0}}{E_{n_0-1}} \right \rfloor  \right ) \right) \left(  E_{n_0} - \left \lfloor \frac{E_{n_0}}{E_{n_0-1}} \right \rfloor E_{n_0-1}     \right )  $$

$$= d E_{n_0-1}  \left(1 -  \left ( \frac{E_{n_0}}{E_{n_0-1}} - \left \lfloor \frac{E_{n_0}}{E_{n_0-1}} \right \rfloor  \right ) \right) \left(  \frac{E_{n_0}}{E_{n_0-1}}- \left \lfloor \frac{E_{n_0}}{E_{n_0-1}} \right \rfloor     \right )  \leq \frac{d E_{n_0-1} }{4} $$
tuples from $F$ in order to gaurantee that those remaining can be proved to be a linearly independent set of columns via the proof in Lemma \ref{k1} as in Eq. \ref{collapse}.  The final inequality above follows from the fact that $ \left(  \frac{E_{n_0}}{E_{n_0-1}}- \left \lfloor \frac{E_{n_0}}{E_{n_0-1}} \right \rfloor     \right )$ is a positive number less than 1.

The total number of columns of $C_{p,\alpha_{n_0-1}; i_{n_0+1}, \alpha_{n_0} }$ is $d E_{n_0}$.  We have shown that at least $|F| = d E_{n_0} - \frac{d E_{n_0-1} }{4}$ of those columns are linearly independent.  Thus the dimension of the kernel of $C_{p,\alpha_{n_0-1}; i_{n_0+1}, \alpha_{n_0} }$ is at most $\frac{d E_{n_0-1} }{4}$.  That is, we now have the bound $E_{n_0+1} \leq \frac{d E_{n_0-1} }{4}$.    

Now, we know that $\left \lceil \frac{E_{n_0}}{E_{n_0-1}} \right \rceil > \frac{r}{d} \geq 1$, and it follows that

$$\frac{E_{n_0-1}}{E_{n_0}} \leq \frac{1}{\left \lfloor \frac{E_{n_0}}{E_{n_0-1}} \right \rfloor - 1} \leq \frac{2}{\left \lfloor \frac{E_{n_0}}{E_{n_0-1}} \right \rfloor} < \frac{2d}{r}$$

Thus,

$$\frac{E_{n_0+1}}{E_{n_0}} \leq   \frac{d E_{n_0-1}}{4 E_{n_0}}     < \frac{d^2}{2r}    $$

So,

$$   \left \lceil \frac{E_{n_0+1}}{E_{n_0}} \right \rceil \leq \left \lceil \frac{d^2}{2r} \right \rceil  \leq \frac{r}{d}   $$

where the final inequality follows by assumption.  Applying Lemma \ref{k1} now gives $E_{n_0+2} = 0$, and we are done.

\end{proof}

\section{Appendix}

We now prove that the constraint matrix $C$ in Eq. \eqref{eq:Cmatrix}
is generically full rank; i.e., with probability $1$. The proof given here
is a generalization of that given in \cite{RS} for qudit chains.
Just as in that earlier proof, we use the ``geometrization theorem''
of \cite{L} to prove full rankness by finding a single set of projectors,
$|v_{n^{(l)},(n+1)^{(h)}}^{p}\rangle$, for which $C$ is full rank. This will be sufficient to prove that $C$ will be full rank with probability one if the projectors are picked at random.

For simplicity we will assume that $k\leq \frac{d}{2}$. It is not clear whether
this is necessary for existence of unfrustrated ground states. However,
if the tree with $k\le \frac{d}{2}$ is frustrated then a larger tree with the same parameters except $k> \frac{d}{2}$ will also be frustrated since it contains a subtree with $k\le \frac{d}{2}$. 
We assume for simplicity that $d$ is even.

The example used to prove full rankness for $k=1$  in \cite{RS} involves an
inductive process by which certain entries of the $\Gamma$ matrices can be found
explicitly. To gain the additional flexibility
needed to prove full-rankness of $C$ for $k>1$ we will introduce
a new technique using the idea that we can, WLOG,
take invertible linear combinations of the $\Gamma_{\alpha_{n-1}^{(1\cdots k)};\alpha_{n}^{(l)}}^{i_{n}^{(l)},[n]}$.
Considering the $\Gamma$'s to be a set of vectors indexed by $\alpha_{n}^{(l)}$,
this is equivalent to taking an invertible change of basis for
the $|\psi_{\alpha_{n}}^{(l)}\rangle$ which does not change the ground space.

Viewing $\Gamma_{\alpha_{n-1}^{(1\cdots k)};\alpha_{n}^{(l)}}^{i_{n}^{(l)},[n]}$ as a $dD_{n-1}^k \times D_n$ matrix with $D_n$ independent columns indexed by $1\le\alpha_{n}^{\left(l\right)}\le D_n$ , let $A$ be an invertible linear map on $D_{n}$
vectors of $\Gamma_{\alpha_{n-1}^{(1\cdots k)};\alpha_{n}^{(l)}}^{i_{n}^{(l)},[n]}$. Then $A$ induces the map $\mathbb{I}_{d}\otimes\mathbb{I}_{D_{n}^{k-1}}\otimes A$ on $C$ by $C\rightarrow C\left(\mathbb{I}_{d}\otimes\mathbb{I}_{D_{n}^{k-1}}\otimes A\right)$.
Since $\mathbb{I}_{d}\otimes\mathbb{I}_{D_{n}^{k-1}}\otimes A$ is invertible, the rank of $C$ is
preserved under this transformation. $ $

We will use this fact in order to run Gaussian elimination on the
$\Gamma_{\alpha_{n-1}^{(1\cdots k)};\alpha_{n}^{(l)}}^{i_{n}^{(l)},[n]}$
and thereby specify certain entries explicitly; more entries than would be attainable using the proof
in \cite{RS}. 
\begin{lem}
\label{columnranklemma-1} For any fixed $l$ consider the $dD_{n-1}^{k}\times D_{n}$
matrix $\mathcal{M}_{(i_{n}^{(l)},\alpha_{n-1}^{(1\cdots k)});\beta_{n}^{(l)}}^{\left(l\right)}\equiv\Gamma_{\alpha_{n-1}^{(1\cdots k)};\beta_{n}^{(l)}}^{i_{n}^{(l)},[n]}$.
Then any $s\times D_{n}$ sub-matrix $W$ in $\mathcal{M}$, with
$s\le D_{n}$ and $i_{n}^{(l)}\leq\frac{d}{2}$ , has rank $s$ with
probability $1$.\end{lem}
\begin{proof}
{} By the argument in the ``geometrization theorem'' of \cite{L},
it is sufficient to prove this statement for a specific choice of
projectors $|v_{(n-1)^{(m)},n^{(l)}}^{p}\rangle$. We will choose
projectors such that

\[
\langle v_{(n-1)^{(m)},(n)^{(l)}}^{p}|i_{n-1}^{(m)},i_{n}^{(l)}\rangle=0,\quad\mbox{for }i_{n}^{(l)}\leq\frac{d}{2}.
\]

From Eq. \eqref{constraint1} the constraint on $\Gamma_{\alpha_{n-1}^{(1\cdots k)};\alpha_{n}^{(l)}}^{i_{n}^{(l)}}$
has the form (see Figure \ref{fig:The-tree-structure})

\begin{align}
 & \langle v_{(n-1)^{(m)},n^{(l)}}^{p}|i_{n-1}^{(m)},i_{n}^{(l)}\rangle\Gamma_{\alpha_{n-2}^{(1\cdots k)};\alpha_{n-1}^{(m)}}^{i_{n-1}^{(m)},[n-1]}\Gamma_{\alpha_{n-1}^{(1\cdots k)};\alpha_{n}^{(l)}}^{i_{n}^{(l)},[n]}=0.
\end{align}
 It thus follows from our choice of projectors that $\Gamma_{\alpha_{n-1}^{(1\cdots k)};\alpha_{n}^{(l)}}^{i_{n}^{(l)},[n]}$
is unconstrained when $i_{n}^{(l)}\leq\frac{d}{2}$. Since $s\leq D_{n}$
, we may choose $\Gamma_{\alpha_{n-1}^{(1\cdots k)};\alpha_{n}^{(l)}}^{i_{n}^{(l)}}$
such that $W$ has the maximum possible rank $s$.

\end{proof}
Given Lemma \ref{columnranklemma-1} we can reduce $W$ to row echelon
form using column operations. The process of Gaussian elimination
would not change the rank of $\Gamma_{\alpha_{n-1}^{(1\cdots k)};\alpha_{n}^{(l)}}^{i_{n}^{(l)},[n]}$.
This process will produce a new set of rows with $s$ pivots. Let
$S$ be the set of indices that index rows of $W$. Equivalently,
the Gaussian elimination produces a new set of $\Gamma_{\alpha_{n-1}^{(1\cdots k)};\alpha_{n}^{(l)}}^{i_{n}^{(l)},[n]}$
such that for every row indexed by $(i,\alpha)\in S$ there exists
a value of $\beta$ such that

\begin{eqnarray}
\Gamma_{\alpha_{n-1}^{(1\cdots k)};\beta}^{i_{n}^{(l)},[n]} & = & 1\text{ if }(i_{n}^{(l)},\alpha_{n-1}^{(1\cdots k)})=(i,\alpha)\label{Eq:gammadiagonalization_}\\
\Gamma_{\alpha_{n-1}^{(1\cdots k)};\beta}^{i_{n}^{(l)},[n]} & = & 0\text{ otherwise}\nonumber 
\end{eqnarray}

\bigskip{}

In order to prove that the constraint matrix $C$ is full rank we
need to prove that

\begin{eqnarray}
 &  & \sum_{p,l,\alpha_{n-1}^{(1\cdots k)},\beta_{n}^{(1)},\cdots,\widehat{\beta_{n}^{(l)}},\cdots,\beta_{n}^{(k)}}y_{p,l,\alpha_{n-1}^{(1\cdots k)},\beta_{n}^{(1)},\cdots,\widehat{\beta_{n}^{(l)}},\cdots,\beta_{n}^{(k)}}C_{p,l,\alpha_{n-1}^{(1\cdots k)},\beta_{n}^{(1)},\cdots,\widehat{\beta_{n}^{(l)}},\cdots,\beta_{n}^{(k)};i_{n+1}^{(h)},\alpha_{n}^{(1\cdots k)}}\\
 & = & \sum_{p,l,\alpha_{n-1}^{(1\cdots k)},\beta_{n}^{(1)},\cdots,\widehat{\beta_{n}^{(l)}},\cdots,\beta_{n}^{(k)}}y_{p,l,\alpha_{n-1}^{(1\cdots k)},\beta_{n}^{(1)},\cdots,\widehat{\beta_{n}^{(l)}},\cdots,\beta_{n}^{(k)}}\prod_{j=1,\cdots,\hat{l,}\cdots,k}\delta\left(\beta_{n}^{j}=\alpha_{n}^{j}\right)\left(\sum_{i_{n}^{(l)}}\langle v_{n^{(l)},(n+1)^{(h)}}^{p}|i_{n}^{(l)},i_{n+1}^{(h)}\rangle\Gamma_{\alpha_{n-1}^{(1\cdots k)};\alpha_{n}^{(l)}}^{i_{n}^{(l)},[n]}\right)\nonumber \\
 & = & 0\mbox{ }\forall i_{n+1}^{(h)},\alpha_{n}^{(1\cdots k)}\nonumber \\
 & \Leftrightarrow & y_{p,l,\alpha_{n-1}^{(1\cdots k)},\beta_{n}^{(1)},\cdots,\widehat{\beta_{n}^{(l)}},\cdots,\beta_{n}^{(k)}},\mbox{ }\forall p,l,\alpha_{n-1}^{(1\cdots k)},\beta_{n}^{(1)},\cdots,\widehat{\beta_{n}^{(l)}},\cdots,\beta_{n}^{(k)}.\nonumber 
\end{eqnarray}

We prove this by showing it for a specific choice of the projectors. We assign the projectors as follows:

\begin{eqnarray}
\langle v_{n^{(l)},(n+1)^{(h)}}^{p}|i_{n}^{(l)},i_{n+1}^{(h)}\rangle= \left\{
\begin{array}{c l}      
    1 & \text{ if } i_{n+1}^{(h)} =  \left ( \left \lceil \frac{2r}{d} \right \rceil +1 \right )(l-1) + \left \lfloor \frac{2p}{d} \right \rfloor \text{ and } i_{n}^{(l)} = p - \left \lfloor \frac{2p}{d} \right \rfloor  \frac{d}{2} \\ 
    0 & \text{ otherwise } .    
\end{array}\right. 
\label{projector_assignment}
\end{eqnarray}

Thus, the projectors are orthogonal basis vectors on the $d^{2}$ dimensional space in the computational basis. Note that this assignment obeys the following threes properties:

1)  For every projector we have

\[
\langle v_{n^{(l)},(n+1)^{(h)}}^{p}|i_{n}^{(l)},i_{n+1}^{(h)}\rangle=0
\]
when $i_{n}^{(l)}>\frac{d}{2}$.  Indeed $p - \left \lfloor \frac{2p}{d} \right \rfloor  \frac{d}{2}  \leq \frac{d}{2}$  is true for all $p$ because it is  the remainder of $\frac{p}{d/2}$.

2)  Furthermore, for each fixed value of $i_{n+1}^{(h)}$,  

\[
\langle v_{n^{(l)},(n+1)^{(h)}}^{p}|i_{n}^{(l)},i_{n+1}^{(h)}\rangle=1
\]
for at most one value of $l$ (but possibly multiple values of $i_{n}^{(l)}$,
and $p$).

3)   Finally, for each $l,p$

\[
\langle v_{n^{(l)},(n+1)^{(h)}}^{p}|i_{n}^{(l)},i_{n+1}^{(h)}\rangle=1
\]
 for exactly one fixed tuple of values $i_{n}^{(l)},i_{n+1}^{(h)}$.

It follows from Eq. (\ref{projector_assignment}) that each projector has at most one non-zero entry. To prove that each projector has {\it exactly} one non-zero entry it remains to verify the third requirement.  We must show that all vectors created this way are non-zero.  This is true iff $\left ( \left \lceil \frac{2r}{d} \right \rceil +1 \right )(l-1) + \left \lfloor \frac{2p}{d} \right \rfloor \leq d$ for all $l$ and $p$.  Since this expression is an increasing function of $l$ and $p$, it is sufficient to show  $\left ( \left \lceil \frac{2r}{d} \right \rceil +1 \right )(k-1) + \left \lfloor \frac{2r}{d} \right \rfloor \leq d$.  To prove this we assume $k \leq \frac{d}{3}$, and $d > 2$.  For $d = 2$ we have $rk \leq \frac{d^2}{4} = 1$, so $r = k = 1$, and the inequality $\left ( \left \lceil \frac{2r}{d} \right \rceil +1 \right )(k-1) + \left \lfloor \frac{2r}{d} \right \rfloor =1\leq d=2$.

We write  $r = x \frac{d}{2} + a$, where $x = \left \lfloor \frac{2 r}{d} \right \rfloor$ and the remainder $0\le a< \frac{d}{2}$,

\begin{align*}
rk = (x \frac{d}{2} +a) k \leq \frac{d^2}{4}
\end{align*}

So, 

\begin{align*}
k \leq \frac{d^2}{4 (x \frac{d}{2} +a)} < \frac{d^2 }{2 x d} = \frac{d}{2 x }
\end{align*}

and thus 

\begin{align*}
\left(  \left \lceil \frac{2r}{d} \right \rceil +1  \right)(k-1) + \left \lfloor \frac{2 r}{d} \right \rfloor  \le (x+2)(k-1) + x = xk + 2(k-1) < x \frac{d}{2x } +2( \frac{ d}{2 x } - 1) = \frac{d}{2} (1 + \frac{2}{x}) - 2
\end{align*}

In the case $x \geq 2$ this gives 

\begin{align*}
\left(  \left \lceil \frac{2r}{d} \right \rceil +1  \right)(k-1) + \left \lfloor \frac{2 r}{d} \right \rfloor    <  \frac{d}{2} (1 + \frac{2}{x}) - 2 \le d-2
\end{align*}

In the case $x < 2$ 

\begin{align*}
\left(  \left \lceil \frac{2r}{d} \right \rceil +1  \right)(k-1) + \left \lfloor \frac{2 r}{d} \right \rfloor   \leq  \left (  \left \lceil \frac{2r}{d} \right \rceil  +1 \right) k \le (x+2) k \leq 3 k \leq d .
\end{align*}

This proves the third assertion.

Now we may suppose that we have performed the Gaussian elimination
described above where the set $S$ is the set $S\equiv\{0,\cdots,\frac{d}{2}\}\times\{0,\cdots.,D_{n-1}\}^{k}$.
Note that $|S|=\frac{d}{2}D_{n-1}^{k}\leq\gamma_{n}D_{n-1}^{k}=D_{n}$
as follows from the work in Section \ref{sec:Recursion-Analysis},
and the fact that $rk\leq\frac{d^{2}}{2}$. This allows us to apply
Lemma \ref{columnranklemma-1} and use Gaussian elimination.

We will therefore assume that $\Gamma$'s have the form described
in \eqref{Eq:gammadiagonalization_} for all $i_{n}^{(l)}\leq\frac{d}{2}$.
Now let us imagine that there are real numbers $y_{p,l,\alpha_{n-1}^{(1\cdots k)},\beta_{n}^{j\neq l}}$
such that

\begin{align}
& \sum_{p,l,\alpha_{n-1}^{(1\cdots k)},\beta_{n}^{(1)},\cdots,\widehat{\beta_{n}^{(l)}},\cdots,\beta_{n}^{(k)}}y_{p,l,\alpha_{n-1}^{(1\cdots k)},\beta_{n}^{(1)},\cdots,\widehat{\beta_{n}^{(l)}},\cdots,\beta_{n}^{(k)}}C_{p,l,\alpha_{n-1}^{(1\cdots k)},\beta_{n}^{(1)},\cdots,\widehat{\beta_{n}^{(l)}},\cdots,\beta_{n}^{(k)};i_{n+1}^{(h)},\alpha_{n}^{(1\cdots k)}}  =\label{eq:linearindependence}\\
 & \sum_{p,l,\alpha_{n-1}^{(1\cdots k)},\beta_{n}^{(1)},\cdots,\widehat{\beta_{n}^{(l)}},\cdots,\beta_{n}^{(k)}}y_{p,l,\alpha_{n-1}^{(1\cdots k)},\beta_{n}^{(1)},\cdots,\widehat{\beta_{n}^{(l)}},\cdots,\beta_{n}^{(k)}}\prod_{j=1,\cdots,\hat{l,}\cdots,k}\delta\left(\beta_{n}^{j}=\alpha_{n}^{j}\right)\left(\sum_{i_{n}^{(l)}}\langle v_{n^{(l)},(n+1)^{(h)}}^{p}|i_{n}^{(l)},i_{n+1}^{(h)}\rangle\Gamma_{\alpha_{n-1}^{(1\cdots k)};\alpha_{n}^{(l)}}^{i_{n}^{(l)},[n]}\right) = 0\nonumber \\
& \\
& \forall i_{n+1}^{(h)},\alpha_{n}^{(1\cdots k)}  .\nonumber 
\end{align}

Take $p',l',\gamma^{(1\cdots k)},\tau^{(j\neq l)}$ to be any fixed
set of values for $p,l,\alpha_{n-1}^{(1\cdots k)},\beta_{n}^{j\neq l}$.
We now prove that $y_{p',l',\gamma^{(1\cdots k)},\tau^{(j\neq l)}}=0$,
thereby completing the proof that $C$ has full row rank.

From above

\[
\langle v_{n^{(l')},(n+1)^{(h)}}^{p'}|i_{n}^{(l')},i_{n+1}^{(h)}\rangle=1
\]
 for exactly one value of $(i_{n}^{(l')},i_{n+1}^{(h)})$, which we
denote by $(i',j')$ (and that it is zero elsewhere). Furthermore,
we know that there is a value of $\beta$ such that

\begin{align}
\Gamma_{\alpha_{n-1}^{(1\cdots k)};\beta}^{i_{n}^{(l)},[n]} & =1\text{ if }(i_{n}^{(l)},\alpha_{n-1}^{(1\cdots k)})=(i',\gamma^{(1\cdots k)})\label{gammadiagonalization}\\
\Gamma_{\alpha_{n-1}^{(1\cdots k)};\beta}^{i_{n}^{(l)},[n]} & =0\text{ otherwise}.\nonumber 
\end{align}

Now, evaluating the Eq. \eqref{eq:linearindependence} at $i_{n+1}^{(h)}=j'$,
and where $\alpha_{n}^{(1\cdots k)}$ is specified by $\alpha_{n}^{j\neq l'}=\tau^{(j\neq l')}$
and $\tau^{(j\neq l')}=\beta$  the constraints collapse to

\[
y_{p',l',\gamma^{(1\cdots k)},\tau^{(j\neq l)}}\langle v_{n^{(l')},(n+1)^{(h)}}^{p'}|i',j'\rangle\Gamma_{\gamma^{(1\cdots k)};\beta}^{i',[n]}=y_{p',l',\gamma^{(1\cdots k)},\tau^{(j\neq l)}}=0.
\]
And so we are done.

\section{Acknowledgements}
We thank Peter W. Shor, Jeffrey Goldstone and Daniel Nagaj for discussions. MC acknowledges the support of NSF IGERT program Interdisciplinary Quantum Information Science and Engineering (iQuISE) through award number 0801525.   RM acknowledges the support of National Science Foundation through grant number CCF-0829421.


\begin{thebibliography}{1}
\bibitem{RS} R. Movassagh, E. Farhi, J. Goldstone, D. Nagaj, and P. W. Shor, Phys. Rev. A 82, 012318 (2010)
\bibitem{AKLT87}I. Affleck, T. Kennedy, E. H. Lieb, and H. Tasaki, Phys. Rev. Lett, 59, 799802 (1987)
\bibitem{Koma95}T. Koma and B. Nachtergaele, Lett. Math. Phys., 40, 1 (1997)
\bibitem{fannes} M. Fannes, B. Nachtergaele, and R. F.Werner, Commun. Math. Phys. 144, 443 (1992)
\bibitem{MPS}D. Perez-Garcia, F. Verstraete, M.M. Wolf, J.I. Cirac, Quantum Inf. Comput. 7, 401 (2007)
\bibitem{Hastings06} M. Hastings, Phys. Rev. B 73, 085115 (2006)
\bibitem{krausMarkov} B. Kraus, H. P. Buechler, S. Diehl, A. Kantian, A. Micheli, and P. Zoller, Phys. Rev. A 78, 042307 (2008)
\bibitem{bravyi} S. Bravyi, arXiv:quant-ph/0602108v1 (2006)
\bibitem{VerstraeteNature} F. Verstraete, M.M. Wolf, and J.I. Cirac, Nature Physics, 5: 633-636, (2009)
\bibitem{RamisThesis} R. Movassagh, Ph.D. Thesis, MIT (April 2012)
\bibitem{alan} A. Edelman, T. A. Arias, S. T. Smith, SIAM. J. Matrix Anal. and Appl., 20(2), 303-353 (1998)
\bibitem{vidal1}G. Vidal, Phys. Rev. Lett. 91, 147902 (2003)
\bibitem{vidal2}G. Vidal, Phys. Rev. Lett. 93, 040502 (2004)
\bibitem{Daniel} D. Nagaj, E. Farhi, J. Goldstone, P. W. Shor, I. Sylvester, Phys. Rev. B 77, 214431 (2008)
\bibitem{vidalTree}Y. Shi, L. Duan, G. Vidal, Phys. Rev. A 74, 022320 (2006)
\bibitem{L} C. Laumann, R. Moessner, A. Scardicchio, and S. L. Sondhi, Quantum Inf. Comput. 10, 0001 (2010)
\end{thebibliography}
\end{document}